\documentclass[preprint,12pt]{elsarticle}
\usepackage{amsfonts}
\usepackage{bbm}
\usepackage{mathrsfs}
\usepackage{hyperref}
\usepackage{color}

\usepackage{amssymb}
\usepackage{bbding}

\usepackage{amsmath}
\usepackage{amssymb}
\usepackage{amsthm}
\usepackage{amsfonts}
\usepackage{amscd}
\usepackage[mathscr]{eucal}

\newcommand{\f}{{\mathbb  F}}
\newcommand{\C}{{\mathcal C}}

\newcommand{\wt}{{\mathrm{wt}}}
\newcommand{\tr}{{\mathrm{Tr}}}

\newtheorem{theorem}{Theorem}
\newtheorem{lemma}[theorem]{Lemma}

\newtheorem{definition}[theorem]{Definition}
\newtheorem{example}{Example}

\makeatletter
 \@addtoreset{equation}{section}
 
\makeatother

\journal{}

\begin{document}

\begin{frontmatter}


\title{ A family of linear codes with three weights\tnoteref{label1} }
\tnotetext[label1]{This work is
supported by a National Key Basic Research Project of China (2011CB302400), National Science Foundation of China (61379139)
and  the ``Strategic Priority Research Program" of the Chinese Academy of Sciences, Grant No. XDA06010701.\\
$ ^{*}\ $Corresponding author.\\ E-mail addresses: yanyang9021@iie.ac.cn(Y. Yan), cczxlf@163.com(F. Li),  wangyan198801@163.com (Q. Wang).}
\author{\small Yang Yan$^{a}$, Fei Li$^{b,*}$, Qiuyan Wang$^{c}$\\
$^a$ National Engineering Laboratory for Information Security Technologies
,  \\ Institute of Information Engineering, Chinese Academy of Sciences, \\
Beijing 100195,  China\\
$^b$  School of Statistics and Applied Mathematics, Anhui University of Finance\\ and Economics,
 Bengbu City,  Anhui Province, {\rm 233000}, China\\
$^c$ School of Computer Science and Software Engineering,  Tianjin Polytechnic \\ University, Tianjin 300387, China}

\begin{abstract}
Recently, linear codes constructed by defining sets have attracted a lot of study, and many optimal linear codes
with a few weights have been produced. The objective of this paper is to present a class of binary linear codes
with three weights. 
\end{abstract}


\end{frontmatter}

\section{Introduction}
\label{}
Throughout this paper,  $q=2^{m}$ for a positive integer $m=2e$, and $h$ denotes  a proper divisor of $m$.  Let $\f_{q}$ and $\f_{q}^{*}$ denote the finite field with  $q$ elements and the multiplicative group of $\f_{q}$, respectively. Let $\f_{2}^{n}$ denote the vector space of all $n-$tuples over $\f_2$.

For $\mathbf{x}\in \f_2^{n}$, the (Hamming) weight $\wt(\mathbf{x})$  is referred to the number of nonzero coordinate in $\mathbf{x}$. The (Hamming) distance $d(\mathbf{x},\mathbf{y})$ between vectors $\mathbf{x},\mathbf{y}\in \f_2^{n}$ is defined to be the weight $\wt(\mathbf{x}-\mathbf{y})$ of $\mathbf{x}-\mathbf{y}\in \f_2^{n}$. The minimum distance of $\C$ is  the least Hamming distance between two distinct codewords. If $\C$ is a $k-$dimensional subspace of $\f_2^{n}$ with minimum distance $d$, then  $\C$ will be called an $[n,k,d]$ linear code.

Let $A_{i}$, also denoted by $A_{i}(\C)$, indicate the number of codewords of weight $i$ in $\C$. The list $A_{i}$ ($0\leq i\leq n$) is called the weight distribution of $\C$. This distribution is usually recorded as the coefficients of a polynomial, the weight enumerator, which is defined as
$$
W_{\C}(x)=\sum_{\mathbf{c}\in C}x^{\wt(\mathbf{c})}=\sum_{i=0}^{n}A_{i}x^{i}.
$$
Clearly, the weight distribution derives   the minimum distance, and thus the error correction capability. In this connection, it would be useful in determining as much as one can about the weight distribution of a specific code or a family of a code \cite{BM72,CKNC12,CW84,S12,V12,XLZD}. If the number of nonzero $A_{i}$ in the sequence $(A_1, A_2, \cdots, A_{n})$ equals $t$, then $\C$ is called a   $t-$weight code.

 Let $D=\{d_1,d_2,\cdots, d_n\}\subseteq \f_{q}$ and $\tr$ denote the trace function from $\f_{q}$ onto $\f_2$. Then a binary linear code associated with $D$ can be defined as
 \begin{equation}\label{defcode}
 \C_{D}=\left\{\left(\tr(xd_1), \tr(xd_2),\cdots, \tr(xd_n)
 \right):x\in \f_{q}\right\}.
 \end{equation}
 The set $D$ is said to be the defining set of $\C_D$. This construction of linear codes was proposed by Ding \cite{D09,D15,DLN08,D07} and has attracted a lot of attention in the last nine years. This approach is generic as many classes of constant, $2$-weight and $3$-weight linear codes were produced  \cite{DD14,DD15,DGZ13,DY13,QDX15,ZD14,ZLFH15}.

 For $a\in \f_2$, the defining set $D_{a}$ is given by
\begin{equation*}
D_{a}=\left\{x\in\f_{q}^{*}:\tr\left(x^{2^{h}+1}+x\right)=a\right\}.
\end{equation*}
Obviously, if $a=0$ and $h=1$, then
$$
D_{0}=\left\{x\in \f_{q}^{*}:\tr\left(x^3+x\right)=0\right\},
$$
that is the same as the defining set of $\C_{D}$ depicted in \cite{Xiang}. Thus the code $\C_{D_{a}}$ presented in this paper is a generalization of that in \cite{Xiang}.

In this paper, we will study the weight distribution of $\C_{D_{a}}$ of \eqref{defcode}. As it turns out that $\C_{a}$ is a linear code with $3$ weights, if $m/h>2$. This implies that this kind of linear codes may be of use in secret sharing schemes \cite{DY13}. We should mention that the idea of solving the weight distribution of $\C_{D_{a}}$ is inspired by that in \cite{Xiang}, which was well-written and inspiring.

The rest of this paper is organized as follows. In Section $2$, we  introduce some notations and  basic results of exponential sums over finite fields, and present some useful conclusions in \cite{C99}.  Section $3$ is devoted to investigating  the weight distribution of the linear codes defined in \eqref{defcode}. In Section $4$, we conclude this paper, and make some comments on the applications of these codes in secret sharing schemes.
\section{Preliminaries}
In this section, we firstly  introduce the definition of  group characters. Then we present some lemmas, which are necessary in proving the main results.

Let $G$ be a finite abelian group. A \textit{character} $\chi$ of $G$ is a homomorphism from $G$ into
the multiplicative group $U$ of complex numbers of absolute value $1$--that is, a mapping from $G$ into $U$ with
$$
\chi(g_1g_2)=\chi(g_1)\chi(g_2)
$$
for all $g_1$, $g_2\in G$ \cite{LN97}. The character $\chi_0$ defined by $\chi_0(g)=1$ for all $g\in G$ is called the \textit{trivial} character of $G$. All other characters of $G$ are called \textit{nontrivial}.

The following lemma present a basic property of characters.
\begin{lemma}[\cite{LN97}, Theorem 5.4]\label{lem1}
If $\chi$ is a nontrivial character of the finite abelian group, then
$$
\sum_{g\in G}\chi(g)=0.
$$
\end{lemma}

In a finite field there are two finite abelian groups--namely, the additive group and the multiplicative group of $\f_q$.   Characters of the \textit{additive group} are called  \textit{additive characters}  of $\f_q$.

To present additive characters of a finite field, we need to introduce an important mapping from $\f_{2^a}$ to $\f_{2^b}$, where $a$ and $b$ denote two integers and satisfy  $b\mid a$.
\begin{definition}[\cite{LN97}, Definition 2.22]
For $\alpha\in \f_{2^a}$, the \textit{trace} $\tr_b^{a}(\alpha)$ of $\alpha$ over $\f_{2^b}$ is defined by
$$
\tr_{b}^{a}(\alpha)=\alpha+\alpha^{2^b}+\cdots+\alpha^{2^{b(l-1)}},
$$
where $l=a/b$. If $b=1$, then $\tr_1^{a}(\alpha)$ is called the absolute trace of $\alpha$.
\end{definition}
For simplicity, we use $\tr$ to denote the trace function from $\f_{q}$ to $\f_2$ in this paper.   By Theorem $2.23$ in \cite{LN97}, $\tr_{b}^{a}(\alpha+\beta)=\tr_b^{a}(\alpha)+\tr_b^{a}(\beta)$ for all $\alpha,\beta\in \f_{2^a}$, and $\tr_b^{a}$ is a linear transformation from $\f_{2^a}$ onto $\f_{2^b}$, where both $\f_{2^a}$ and $\f_{2^b}$ are viewed as vector apace over $\f_{2^b}$.

By the linearity of   $\tr$, the function $\chi_1$ defined by
$$
\chi_1(c)=e^{2\pi i\tr(c)/2}=e^{\pi i\tr(c)} \ \textrm{for \ all\ } c\in \f_q
$$
is an additive  character of $\f_q$.

For $a,b\in \f_q$, define the following exponential sum
$$
S_{h}(a,b)=\sum_{x\in \f_q}\chi_1\left(x^{2^{h}+1}+bx\right),
$$
where $h$ is a proper divisor of $m$.

In order to determine the weight distribution of the code $\C_{D_a}$ defined in \eqref{defcode}, we need the values of $S_{h}(a,b)$, which are presented in the following lemmas.
\begin{lemma}[\cite{C99}, Theorem 4.1]\label{lem2}
When $m/h$ is odd, we have
$$\sum_{x\in \mathbb{F}_{q}}\chi_1\left(ax^{2^{h}+1}\right)=0 $$
for each $a\in\mathbb{F}_{q}^{*}.$
\end{lemma}
\begin{lemma}[\cite{C99}, Theorem 4.6]\label{lem3.1}
Suppose $m/h$ is odd,  then we have
$$
S_{h}(1,1)={\left(\frac{2}{m/h}\right)}^{h}2^{\frac{m+h}{2}},
$$
where and hereafter $(\frac{\cdot}{\cdot})$ denotes  the Jacobi symbol.
\end{lemma}
\begin{lemma}[\cite{C99}, Theorem 4.2]\label{lem3}
Let $b\in\mathbb{F}_{q}^{*} $ and suppose $m/h$ is odd. Then
$$S_{h}(a,b)=S_{h}\left(1,bc^{-1}\right), $$
where $ c $ is the unique element satisfying
$ c^{2^{h}+1}=a. $ Further we have
$$
S_{h}(1,b)=\left\{\begin{array}{ll}
 0, & \textrm{if\ } \ \tr_{h}^{m}(b)\neq 1, \\
 \pm 2^{\frac{m+h}{2}}, & \textrm{if\ } \ \tr_{h}^{m}(b)= 1,
 \end{array}
 \right.
$$

\end{lemma}
\begin{lemma}[\cite{C99}, Theorem 5.2]\label{lem4}
Let $m/h$ be even so that $m=2e$ for some integer $e$. Then
$$
S_{h}(a,0)=\left\{\begin{array}{ll}
(-1)^{\frac{e}{h}}2^{e}, & \textrm{if\ } \ a\neq g^{t(2^{h}+1)}  \textrm{\ for any integer\ } t, \\
-(-1)^{\frac{e}{h}}2^{e+h}, & \textrm{if\ } \ a= g^{t(2^{h}+1)}  \textrm{\ for some integer\ } t,
\end{array}
\right.
$$
where $g$ is a generator of $\mathbb{F}_{q}^{*}$.
\end{lemma}
\begin{lemma}[\cite{C99}, Theorem 5.3]\label{lem5}
Let $b\in\mathbb{F}_{q}^{*} $ and suppose $m/h$ is even so that $m=2e $
for some integer $ e.$ Let $ f(x)=a^{2^{h}}x^{2^{2h}}+ax \in \mathbb{F}_{q}[x]. $ There are two cases.\
\begin{enumerate}

\item If $a\neq g^{t(2^{h}+1)}$ for any integer  $t$, then $f$ is a permutation polynomial of
$\mathbb{F}_{q}. $ Let $x_{0}$ be the unique element satisfying $f(x)=b^{2^{^{h}}}.$ Then
$$S_{h}(a,b)
=(-1)^{\frac{e}{h}}2^{e}\chi_1\left(ax_{0}^{2^{h}+1}\right).$$
\item If $a= g^{t(2^{h}+1)}$ for some integer $t$, then $S_{h}(a,b)=0$ unless the equation
$f(x)=b^{2^{^{h}}} $ is solvable. If this equation is solvable, with solution $ x_{0}$ say, then
$$
S_{h}(a,b)
=\left\{\begin{array}{ll}
-(-1)^{\frac{e}{h}}2^{e+h}\chi_1\left(ax_{0}^{2^{h}+1}\right), & \textrm{if\ } \ \tr_{h}^{m}(a)= 0, \\
(-1)^{\frac{e}{h}}2^{e}\chi_1\left(ax_{0}^{2^{h}+1}\right),  & \textrm{if\ } \ \tr_{h}^{m}(a)\neq 0.
\end{array}
\right.
$$
\end{enumerate}
\end{lemma}

\begin{lemma}\label{lemma8}
Let symbols be the same as before. Then the equation
\begin{equation}\label{eq3}
x^{2^{2h}}+x=1
\end{equation}
 has solutions in $\f_{q}$ if and only if $m/h\equiv0\pmod{4}$.
\end{lemma}
\begin{proof}
Sufficiency: Suppose $m/h\equiv0\pmod{4}$. Note that
$$
x^{2^{2h}}+x=\tr_{2h}^{4h}(x) \textrm{ for\ all } x\in \f_{2^{4h}}.
$$
The equation \eqref{eq3} has solutions in $\f_{2^{4h}}$, as $\tr_{2h}^{4h}$ is  surjective  from $\f_{2^{4h}}$ to $\f_{2^{2h}}$. Then the desired result follows from the fact that $\f_{2^{2h}}\subseteq \f_{2^m}$.

Necessity: Suppose $x_0\in \f_q$ is a solution of \eqref{eq3}. If $m/h\not\equiv0\pmod{4}$, then $m/h\equiv2\pmod{4}$ or $m/h\equiv1\pmod{2}$. For the case  $m/h\equiv2\pmod{4}$, we have
\begin{equation}\label{eq2}
\tr_{2h}^{m}\left(x^{2^{2h}}+x_0\right)=2\tr_{2h}^{m}(x_0)=0.
\end{equation}
 On the other hand, we obtain
 $$\tr_{2h}^{m}(x_0^{2^{2h}}+x_0)=\tr_{2h}^{m}(1)=m/(2h)=1,
 $$
 which contradicts with \eqref{eq2}.

 If  $m/h\equiv1\pmod{2}$,  we get  $(2m)/h\equiv2\pmod{4}$.  By the discussion  above, we know that the equation $x^{2^{2h}}+x=1$ has no solution in $\f_{2^{2m}}$. So does it in $\f_{q}=\f_{2^{m}}$, since $\f_{q}\subseteq \f_{2^{2m}}$.
\end{proof}
\begin{lemma}\label{lemma9}
Let $m/h\geq2$ be even. Then
$$
S_{h}(1,1)=\left\{\begin{array}{ll}
                    0, & \textrm{if $m/h$} \equiv2\pmod{4}, \\
                    -(-1)^{\frac{m}{4}}2^{e+h}, & \textrm{if $m/h$} \equiv0\pmod{4}.
                  \end{array}
                  \right.
$$
\end{lemma}
\begin{proof}
Note that $\tr_{h}^{m}(1)=0$, as $m/h$ is even. It then follows from Lemmas \ref{lem5} and \ref{lemma8} that
\begin{align*}
S_{h}(1,1)&=\left\{\begin{array}{ll}
                    0, & \textrm{if $m/h$}\equiv2\pmod{4}, \\
                    -2^{e+h}\chi_0\left(x_0^{2^h+1}\right), & \textrm{if $m/h$}\equiv0\pmod{4},
                  \end{array}
                  \right.
\end{align*}
where $x_0$ is a solution of the equation $x^{2^{2h}}+x=1$, when $m/h\equiv0\pmod{4}$. By the proof of Lemma \ref{lemma8}, we know that the solutions of \eqref{eq3} are just the preimage of $\tr_{2h}^{4h}(1)$ and obviously they are in $\f_{2^{4h}}$. If $m/(4h)\equiv0\pmod{2}$, then $m/4\equiv0\pmod{2}$ and
$$
\tr\left(x_0^{2^h+1}\right)=\frac{m}{4h}\tr_{1}^{4h}\left(x_0^{2^h+1}\right)=0, \textrm{i.e., } (-1)^{\tr\left(x_0^{2^h+1}\right)}=(-1)^{\frac{m}{4}}.
$$
If $m/(4h)\equiv1\pmod{2}$, then $m/4\equiv h\pmod{2}$ and
\begin{align}\label{eq4}
\tr\left(x_0^{2^h+1}\right)&=\frac{m}{4h}\tr_{1}^{4h}\left(x_0^{2^h+1}\right)\nonumber\\
&=\tr_{1}^{4h}\left(x_0^{2^{h}+1}\right)\nonumber \\
&=\tr_{1}^{2h}\left(\tr_{2h}^{4h}\left(x_0^{2^h+1}\right)\right).
\end{align}
By the definition of the trace function and  $x_0$, we have
\begin{align}\label{eq5}
\tr_{1}^{2h}\left(\tr_{2h}^{4h}\left(x_0^{2^h+1}\right)\right)&=\tr_1^{2h}\left(x_0^{2^h+1}+x_0^{2^{2h}(2^h+1)}\right)\nonumber \\
&=\tr_{1}^{2h}\left(x_0^{2^h+1}+(1+x_0)^{2^h+1}\right)\nonumber \\
&=\tr_1^{2h}\left(x_0^{2^{h}}+x_0+1\right)\nonumber \\
&=\tr_{1}^{h}\left(\tr_{h}^{2h}\left(x_0^{2^h}+x_0+1\right)\right)\nonumber \\
&=\tr_{h/1}(1)\nonumber\\
&=h.
\end{align}
Combining \eqref{eq4} and \eqref{eq5}, we have
$$\tr\left(x_0^{2^h+1}\right)=h, \textrm{ i.e., }   (-1)^{\tr\left(x_0^{2^h+1}\right)}=(-1)^{\frac{m}{4}}.$$
Hence, we get
$$
S_{h}(1,1)=\left\{\begin{array}{ll}
                    0,  & \textrm{if $m/h$}\equiv2\pmod{4}, \\
                    -(-1)^{\frac{m}{4}}2^{e+h}, & \textrm{if $m/h$}\equiv0\pmod{4}.
                  \end{array}
                  \right.
$$
This completes the proof of this lemma.
\end{proof}

\section{Results and Proofs}
In this section, we will determine the parameters of the code $\C_{D_a}$ $(a\in \f_2)$ defined in \eqref{defcode},   and give the proofs of these parameters.

Define
$$
N_{a}=\left|\left\{x\in \mathbb{F}_{q}: \tr\left(x^{2^{h}+1}+x\right)=a\right\}\right|.
$$
By definition, the length $n_a$ of the code $\C_{D_a}$ satisfies
\begin{equation}\label{eq1.1}
n_a=N_a+a-1.
\end{equation}
 It can be easily checked that
\begin{align}\label{eq-weight1}
N_{a} &= 2^{-1}\sum_{x\in \mathbb{F}_{q}}\sum_{y \in \f_{2}}(-1)^{y\tr\left(x^{2^{h}+1}+x\right)-ya}
  \nonumber \\
 &=2^{-1}\sum_{x\in \mathbb{F}_{q}}\left(1+(-1)^{\tr\left(x^{2^{h}+1}+x\right)-a}\right)
\nonumber \\
&= 2^{m-1}
+ 2^{-1}\sum_{x\in \mathbb{F}_{q}}(-1)^{\tr\left(x^{2^{h}+1}+x\right)-a}  \nonumber\\
&= 2^{m-1} + 2^{-1}(-1)^{a}S_{h}(1,1)
\end{align}

Define
$$N(a,b)=\left|\left\{x\in \mathbb{F}_{q}: \tr\left(x^{2^{h}+1}+x\right)=a \textrm{ and } \tr(bx)=0\right\}\right|.$$
Denote $\wt(\mathbf{c}_b)$ the Hamming weight of the  codeword $
 \mathbf{c}_b$ with $b\in \mathbb{F}_{q}^{*}$
 of the code $\C_{D_{a}}$. It is easy to see that
\begin{equation} \label{eq-wt}
\wt(\mathbf{c}_b)=N_a-N(a,b).
\end{equation}
 For any $b\in \f_q^{*}$, by Lemma \ref{lem1}, we have
\begin{align}\label{eq-weight}
N(a,b) &= 2^{-2}\sum_{x\in \mathbb{F}_{q}}\left(\sum_{y \in \f_{2}}(-1)^{y\tr\left(x^{2^{h}+1}+x\right)-ya}\right)
\left(\sum_{z \in \f_{2}}(-1)^{z\tr(bx)}\right)  \nonumber \\
 &=2^{-2}\sum_{x\in \mathbb{F}_{q}}\left(1+(-1)^{\tr\left(x^{2^{h}+1}+x\right)-a}\right)
\left(1+(-1)^{\tr(bx)}\right)\nonumber \\
&= 2^{m-2}
+ 2^{-2}\sum_{x\in \mathbb{F}_{q}}(-1)^{\tr\left(x^{2^{h}+1}+x\right)-a} \nonumber
\quad +
2^{-2}\sum_{x\in \mathbb{F}_{q}}(-1)^{\tr\left(x^{2^{h}+1}+bx+x\right)-a} \nonumber\\
&= 2^{m-2} + 2^{-2}(-1)^{a}\left(S_{h}(1,1)+S_{h}(1,b+1)\right) \nonumber\\
\end{align}

\begin{table}[ht]
\centering
\caption{The weight distribution of the codes of Theorem \ref{theorem1}}\label{tal:weightdistribution1}
\begin{tabular}{|l|l|}
\hline
\textrm{Weight} $w$ \qquad& \textrm{Multiplicity} $A$   \\
\hline
0 \qquad&   1  \\
\hline
$2^{m-2}$ \qquad&  $2^{m-h-1}-1+\left(\frac{2}{m/h}\right)^{h}2^{\frac{m-h-2}{2}}$  \\
\hline
$2^{m-2}+(-1)^{a}\left(\frac{2}{m/h}\right)^{h}2^{\frac{m+h-2}{2}}$  \qquad& $2^{m-h-1}-\left(\frac{2}{m/h}\right)^{h}2^{\frac{m-h-2}{2}}$  \\
\hline
$2^{m-2}+(-1)^{a}\left(\frac{2}{m/h}\right)^{h}2^{\frac{m+h-4}{2}}$ \qquad&    $2^{m}-2^{m-h}$ \\
\hline
\end{tabular}
\end{table}

\begin{theorem}\label{theorem1}
Let $m/h$ be odd. For $a\in \f_2$, the code ${\C}_{D_{a}}$ defined in \eqref{defcode} is a
$\left[n_a, m\right]$ binary
linear code with the weight distribution in $\autoref{tal:weightdistribution1}$, where $n_a=2^{m-1}+(-1)^{a}\left(\frac{2}{m/h}\right)^{h}2^{\frac{m+h-2}{2}}+a-1$.
\end{theorem}
\begin{proof}
 We only prove the conclusions on the parameters of the code $\C_{D_0}$, since the conclusions  on $\C_{D_1}$ can be similarly proved.

By Lemma \ref{lemma9}, \eqref{eq1.1} and \eqref{eq-weight1}, we obtain
\begin{equation}\label{eq-n0}
n_0=N_0-1=2^{m-1}+\left(\frac{2}{m/h}\right)^{h}2^{\frac{m+h-2}{2}}-1.
\end{equation}

For $b\in\f_{q}\setminus \{0,1\}$, by Lemma \ref{lem3}, we have
$$
S_{h}(1,b+1)=\left\{\begin{array}{ll}
                      0, & \textrm{if } \tr(b+1)\neq1, \\
                      \pm 2^{\frac{m+h}{2}}, & \textrm{if } \tr(b+1)=1.
                    \end{array}
                    \right.
$$
For $b=1$, by Lemma \ref{lem2}, $S_{h}(1,b+1)=S_{h}(1,0)=0$. Then for $b\in \f_q^{*}$, we get
 $$S_h(1,b+1)\in \{0,\pm2^{\frac{m+h}{2}}\}.$$
It follows from \eqref{eq-weight} that
$$
N(a,b)\in\left\{\mu, \mu\pm 2^{\frac{m+h}{2}}\right\},
$$
where $\mu=2^{m-2}+2^{-2}\left(\frac{2}{m/h}\right)^{h}2^{\frac{m+h}{2}}$.
Hence, the weight $\wt(\mathbf{c}_b)$ of the codeword $\mathbf{c}_b$ satisfies
\begin{align*}
\wt(\mathbf{c}_b)&=N_0-N(0,b)\\
&\in \left\{2^{m-2}, 2^{m-2}+\left(\frac{2}{m/h}\right)^{h}2^{\frac{m+h-2}{2}}, 2^{m-2}+\left(\frac{2}{m/h}\right)^{h}2^{\frac{m+h-4}{2}}\right\},
\end{align*}
and the code $\C_{D_0}$ has all the three weights in the above set. The minimum distance of the dual code $\C_{D_0}^{\perp}$ of $\C_{D_0}$ is at least $3$, since $0\not\in D_0$ and the elements in $D_0$ are pairwise distinct.
Define
\begin{align*}
\omega_1&=2^{m-2},\\
\omega_2&=2^{m-2}+\left(\frac{2}{m/h}\right)^h2^{\frac{m+h-2}{2}},\\
\omega_3&=2^{m-2}+\left(\frac{2}{m/h}\right)^h2^{\frac{m+h-4}{2}}.
\end{align*}
Next we will calculate the multiplicity $A_{\omega_i}$ of $\omega_i$.
By the computation above, the number  of $b\in \f_{q}^{*}$ such that  $\wt(\mathbf{c}_b)=\omega_3$  is equal to
$$
1+\left|\left\{b\in \f_{q}\setminus \{0,1\}: \tr_{h}^{m}(b+1)\neq1\right\}\right|.
$$
It is known that \cite{LN97}
$$
\left|\left\{b\in \f_{q}\setminus \{0,1\}: \tr_{h}^{m}(b+1)=1\right\}\right|=2^{m-h}-1.
$$
Hence, $A_{\omega_3}=2^{m}-2^{m-h}$. By the first two Pless Power Moment (\cite{HP03}, P. 260), we get the following system of equations
\begin{align}\label{eq-powe}
\left\{\begin{array}{ll}
         A_{\omega1}+A_{ \omega2}+A_{ \omega3}=2^{m}-1, \\
         \omega_1A_{ \omega1} + \omega_2A_{\omega_2}+ \omega_3A_{ \omega_3}=n_02^{m-1},
       \end{array}
       \right.
\end{align}
where $n_0$ is given by \eqref{eq-n0}. Solving the system of equations in \eqref{eq-powe}  yields the weight distribution of $\C_{D_0}$. The dimension of $\C_{D_0}$ is $m$, as $\wt(\mathbf{c}_b)>0$ for $b\in \f_q^{*}$. The proof of the parameters of $\C_{D_0}$ is completed.
\end{proof}

\begin{table}[ht]
\centering
\caption{The weight distribution of the codes of Theorem \ref{theorem5}}\label{tal:weightdistribution4}
\begin{tabular}{|l|l|}
\hline
\textrm{Weight} $w$ \qquad& \textrm{Multiplicity} $A$   \\
\hline
0 \qquad&   1  \\
\hline
$2^{m-2}$ \qquad&  $2^{m}-2^{m-2h}-1$  \\
\hline
$2^{m-2}-2^{e+h-2}$  \qquad& $2^{m-2h-1}+(-1)^{a}2^{e-h-1}$  \\
\hline
$2^{m-2}+2^{e+h-2}$ \qquad \qquad \qquad \qquad &    $2^{m-2h-1}-(-1)^{a}2^{e-h-1}$  \qquad \qquad\\
\hline
\end{tabular}
\end{table}

\begin{theorem} \label{theorem5}
Let $2 < m/h\equiv2\pmod4.$ Then the code $ {\C}_{D_{a}}$
defined in \eqref{defcode} is a $[2^{m-1}+a-1, m]$ binary linear code with the weight distribution in $\autoref{tal:weightdistribution4}$.
\end{theorem}
\begin{proof}
Similarly, we only prove the results of $\C_{D_0}$.

It follows from Lemma \ref{lemma9}, \eqref{eq1.1} and \eqref{eq-weight1} that the length $n_0$ of
$\C_{D_0}$ is given by
$$
n_0=N_0-1=2^{m-1}-1.
$$

Note that $\tr_{h}^{m}(1)=0$, as $m/h$ is even.  For $b\in\f_{q}^{*}\setminus \{1\}$, by Lemma \ref{lem5}, we know   $S_{h}(1,b+1)\in \{0,\pm2^{e+h}\}$. If $b=1$, by Lemma \ref{lem4} $S_{h}(1,0)=-2^{e+h}$. Hence, for $b\in \f_{q}^{*}$, we have
$$
S_{h}(1,b+1)\in \{0, \pm2^{e+h}\}.
$$
It then follows from Lemma \ref{lemma9} and \eqref{eq-weight} that
$$
N(0,b)\in \{2^{m-2}, 2^{m-2}\pm2^{e+h-2}\}
$$
 for $b\in \f_{q}^{*}$.
Hence,
$$
\wt(\mathbf{c}_b)\in \{2^{m-2}, 2^{m-2}\pm2^{e+h-2}\},
$$
for all $b\in \f_{q}^{*}$.

As stated in the proof of Theorem \ref{theorem1}, the minimum distance $d^{\perp}$ of the dual code $\C_{D_0}^{\perp}$ of $\C_{D_0}$ satisfies $d^{\perp}\geq3$.

Let $\omega_1=2^{m-2}$, $\omega_2=2^{m-2}+2^{e+h-2}$ and $\omega_3=2^{m-2}-2^{e+h-2}$.  The first three Pless Power Moments (\cite{HP03}, P. 260) give the following system of equations
$$
\left\{\begin{array}{ll}
           A_{\omega1}+A_{ \omega2}+A_{ \omega3}=2^{m}-1, \\
         \omega_1A_{ \omega1} + \omega_2A_{\omega_2}+ \omega_3A_{ \omega_3}=2^{2m-2},\\
         \omega_1^{2}A_{ \omega1} + \omega_2^{2}A_{\omega_2}+ \omega_3^{2}A_{ \omega_3}=\left(2^{m-1}+1\right)2^{2m-3}.
       \end{array}
       \right.
$$
Solving the above system of equations, we get the weight distribution of \autoref{tal:weightdistribution4}.
\end{proof}

\noindent{\textbf{Remark}:}
\begin{enumerate}
\item   If $m/h=2, $ then the code ${\C}_{D_{0}}$
defined in \eqref{defcode} is a $[2^{m-1}-1, m-1,2^{m-2}]$ constant  binary linear code.
And the code ${\C}_{D_{0}}\backslash \{0\}$ is an optimal constant-weight code with
Johnson Bound I. 

\item If $m/h=2, $ then the code ${\C}_{D_{1}}$ is a $[2^{m-1}, m]$ binary linear code with the weight distribution in  $\autoref{tal:weightdistribution3}$.  
\end{enumerate}

For the case $m/h=2$, we omit  the proofs of the parameters of $\C_{D_a}$ $(a\in \f_2)$ and the details are left to the reader.
\begin{table}
\centering
\caption{The weight distribution of  $\C_{D_1}$, when $m/h=2$}\label{tal:weightdistribution3}
\begin{tabular}{|l|l|}
\hline
\textrm{Weight} $w$ \qquad& \textrm{Multiplicity} $A$   \\
\hline
0 \qquad&   1  \\
\hline
$2^{m-1}$ \qquad&  $1$  \\
\hline
$2^{m-2}$ \qquad \qquad \qquad \qquad \qquad   & $2^{m}-2$  \qquad \qquad \qquad \qquad \ \\
\hline
\end{tabular}
\end{table}

\begin{theorem} \label{theorem7}
Let $m/h\equiv0\pmod4.$ Then the code $ {\C}_{D_{a}}$
defined in \eqref{defcode} is a $[2^{m-1}-(-1)^{\frac{m+4a}{4}}2^{e+h-1}+a-1, m]$
binary linear code with the weight distribution in $\autoref{tal:weightdistribution6}$.
\end{theorem}
\begin{proof}
The proof of this theorem is
analogous to that in Theorem \ref{theorem5} and will not be included here.
\end{proof}
\begin{table}
\centering
\caption{The weight distribution of the codes of Theorem \ref{theorem7}}\label{tal:weightdistribution6}
\begin{tabular}{|l|l|}
\hline
\textrm{Weight} $w$ \qquad& \textrm{Multiplicity} $A$   \\
\hline
0 \qquad&   1  \\
\hline
$2^{m-2}-(-1)^{\frac{m+4a}{4}}2^{e+h-2}$ \qquad&  $2^{m}-2^{m-2h}$  \\
\hline
$2^{m-2}$  \qquad& $2^{m-2h-1}-(-1)^{\frac{m}{4}}2^{e-h-1}-1$  \\
\hline
$2^{m-2}-(-1)^{\frac{m+4a}{4}}2^{e+h-1}$ \qquad&    $2^{m-2h-1}+(-1)^{\frac{m}{4}}2^{e-h-1}$ \\
\hline
\end{tabular}
\end{table}

\begin{example}
Let $(m,h)=(9,3)$. For $a=0$, the code $\C_{D_0}$ has parameters $[223,9,96]$ and weight distribution
$$
1+36x^{96}+448x^{112}+27x^{128};
$$
for $a=1$, the code $\C_{D_1}$ has parameters $[288,9,128]$ and weight distribution
$$
1+27x^{128}+448x^{144}+36x^{160}.
$$
\end{example}
\begin{example}
Let $(m,h)=(12,2)$. For $a=0$, the code $\C_{D_0}$ has parameters $[2047,12,960]$ and weight distribution
$$
1+136x^{960}+3839x^{1024}+120x^{1088};
$$
for $a=1$, the code $\C_{D_1}$ has parameters $[2048,12,960]$ and weight distribution
$$
1+120x^{960}+3839x^{1024}+136x^{1088}.
$$
\end{example}
\begin{example}
Let $(m,h)=(8,2)$. For $a=0$, the code $\C_{D_0}$ has parameters $[96,832]$ and weight distribution
$$
1+10x^{32}+240x^{48}+5x^{64};
$$
for  $a=1$, the code $\C_{D_1}$ has parameters $[160,8,64]$ and weight distribution
$$
1+5x^{64}+240x^{80}+10x^{96}.
$$
\end{example}

\noindent{References}

\end{document}